\documentclass[12pt]{article}
\usepackage{amsfonts,amssymb,amsmath}
\usepackage[pagebackref]{hyperref}

\numberwithin{equation}{section}
\newtheorem{theorem}{Theorem}[section]

\newtheorem{proposition}[theorem]{Proposition}

\newtheorem{defin}[theorem]{Definition}

\newenvironment{proof}{\noindent \textbf{Proof: }}{\hfill
$\Box$  \vspace{1ex}}
\newenvironment{definition}{\begin{defin}\em}{\end{defin}}
\newtheorem{defins}[theorem]{Definitions}

\newtheorem{exs}[theorem]{Examples}

\newtheorem{ex}[theorem]{Example}

\newtheorem{rem}[theorem]{Remark}

\newtheorem{rems}[theorem]{Remarks}

\def\balpha{\boldsymbol{\alpha}}

\def\fq{\mathbb{F}_q}

\def\N_0{\mathbb{N}_0}

\def\RM{\textrm{RM}}
\def\prm{\textrm{PRM}}

\begin{document}


\begin{center}
{\Large\textbf{On the next-to-minimal weight of projective Reed-Muller codes}}
\end{center}
\vspace{3ex}

\noindent
\begin{center} 
\textsc{C\'{\i}cero  Carvalho and Victor G.L.\ Neumann}\footnote{\noindent 
Authors' emails: cicero@ufu.br and gonzalo@famat.ufu.br. Both authors were 
partially supported by  grants from CNPq and FAPEMIG.} 
\end{center}

\noindent
{\small Faculdade de Matem\'atica,  Universidade Federal de Uberl\^andia,    Av.\ J.\ N.\ \'Avila 2121, 38.408-902 \ \ Uberl\^andia - MG, Brazil}
\vspace{3ex}

\noindent
{\footnotesize \textbf{Abstract.} 
In this paper we present several values for the next-to-minimal weights of 
projective Reed-Muller codes. We work over $\mathbb{F}_q$ with $q \geq 3$ since 
in \cite{carv-neu-ieee} we have determined the complete values for the 
next-to-minimal weights of binary projective Reed-Muller codes. As in 
\cite{carv-neu-ieee} here we also find examples of codewords with 
next-to-minimal weight whose set of zeros is not in a hyperplane arrangement.}
\vspace{2ex}

\section{Introduction} \label{intro}

Reed-Muller codes were introduced in 1954 by 
D.E.\ Muller (\cite{muller}) as codes defined over $\mathbb{F}_2$, and a 
decoding algorithm for them 
was devised by I.S.\ Reed 
(\cite{reed}). In 1968  Kasami, 
Lin, and Peterson (\cite{kasami}) extended the original definition to a finite 
field $\mathbb{F}_q$, where $q$ is any prime power, and named these codes 
``generalized Reed-Muller codes''. They also presented some results on the 
weight 
distribution, the dimension of the codes being determined in later works. In 
coding theory one  is always interested in the values of the higher Hamming 
weights of a code because of their relationship with the code performance, but 
usually this is not a simple problem. For the generalized Reed-Muller codes, 
the complete determination of the second lowest Hamming weight, also called 
next-to-minimal weight,  
was only completed in 2010, when Bruen (\cite{bruen}) observed that the value 
of these weights could be obtained from unpublished results in the Ph.D.\ 
thesis 
of D.\ Erickson (\cite{erickson}) and Bruen's own results from 1992 and 2006. 
Now that Bruen called the attention the Erickson's thesis we know that the 
complete list of the next-to-minimal weights for the generalized Reed-Muller 
codes may be obtained by combining results from Erickson's thesis with results 
by Geil (see \cite{geil}) or with results by  Rolland (see \cite{rolland}).

In 1990 Lachaud introduced the class of projective Reed-Muller codes
(see \cite{lachaud}). The parameters of this codes were determined by Serre 
(\cite{serre}), for some cases, and by S\o rensen (\cite{sorensen}) for the 
general case. As for the determination of the next-to-minimal weight for these 
codes, there are some results (also about higher Hamming 
weights) on this subject by Rodier and Sboui (\cite{rodier-sboui}, 
\cite{rodier-sboui2}, \cite{sboui1}) and also by Ballet and Rolland 
(\cite{rolland-ballet}). Recently the authors of this note completely 
determined  the next-to-minimal weight of projective Reed-Muller codes defined 
over $\mathbb{F}_2$ (see \cite{carv-neu-ieee})
In this paper we present several results on the 
next-to-minimal 
Hamming weights for projective Reed-Muller codes, including the complete 
determination of the next-to-minimal weights for the case of projective 
Reed-Muller codes defined over $\mathbb{F}_3$. In the next section we recall 
the definitions of the 
generalized and projective Reed-Muller codes, and some results of 
geometrical nature that will allow us to determine many cases of higher 
Hamming weights for the projective Reed-Muller codes, which we do in the last 
section. 

\section{Preliminary results}

Let $\fq$ be a finite field and let $I_q = (X_1^q - X_1, \ldots, X_n^q - X_n) 
\subset \fq[X_1, \ldots, X_n]$ be the ideal of polynomials which vanish at all 
points $P_1, \ldots, P_{q^n}$ of the affine space $\mathbb{A}^n(\fq)$. Let 
$\varphi: \fq[X_1, \ldots, X_n]/I_q \rightarrow \fq^{q^n}$ be the $\fq$-linear 
transformation given by $\varphi(g + I_q) = (g(P_1) , \ldots, g(P_{q^n}) )$.

\begin{definition}
Let $d$ be a nonnegative integer.  The generalized Reed-Muller code of order 
$d$ is defined 
as $\RM(n, d) = \{ \varphi(g + I_q) \, | \, g = 0 \textrm{ or } \deg(g) \leq d 
\}$.  
\end{definition}

It is not difficult to prove that $\RM(n, d) = \fq^{q^n}$ if $d \geq n(q - 1)$, 
so in this case 
the minimum distance is 1. Let $d \leq n(q -1)$ and write $d = a (q - 1) + 
b$ with $0 < b \leq q - 1$, then the minimum distance of $\RM(n, d)$ is 
\[
W^{(1)}_{\RM}(n, d) = (q - b) q^{n - a - 1}.
\]
According to \cite{erickson} and \cite{bruen} (see also \cite[Thm.\ 
9]{rolland-ballet})  the next-to-minimal 
 weight  $W^{(2)}_{\RM}(n, d)$ of 
$\RM(n, d)$ is equal to 
\begin{equation} \label{table-2ndwt}
W^{(2)}_{\RM}(n,d) =  W^{(1)}_{\RM}(n, d) + c q^{n-a-2}   = \left(1 + 
\dfrac{c}{(q- b)q}  \right) (q - b)q^{n-a-1} 
\end{equation}
where
\[
c = \left\{ \begin{array}{lcl}  
q &     \textrm{ if } & a = n - 1; \\
b - 1 & \textrm{ if } & a < n - 1 \textrm{ and } 1 < b \leq (q + 
1)/2;            \\
      & \textrm{ or } & a < n - 1 \textrm{ and } b = q - 1  \neq 1;         \\
q      & \textrm{ if } & a = 0 \textrm{ and } b = 1;      \\
q - 1  & \textrm{ if } & q < 4, 0 < a < n - 2, \textrm{ and } b = 1;           
\\
q - 1  & \textrm{ if } & q = 3, 0 < a = n - 2 \textrm{ and } b = 1;          \\
q      & \textrm{ if } & q = 2,    a = n - 2 \textrm{ and } b = 1;           \\
q      & \textrm{ if } & q \geq 4, 0 <  a \leq n - 2 \textrm{ and } b = 
1;           \\
b - 1 & \textrm{ if } & q \geq 4,  a \leq n - 2 \textrm{ and } (q + 1)/2 < b.
\end{array}  \right.
\]

 We will need some specific values 
of $W^{(2)}_{\RM}(n,d)$ in the next section. 	

Let $Q_1,\ldots, Q_N$ be the points of $\mathbb{P}^n(\fq)$, where $N = q^n + 
\ldots + q + 1$. From e.g.\ \cite{idealapp} or 
\cite{mercier-rolland}  we get that the homogeneous ideal $J_q \subset 
\fq[X_0, \ldots 
, X_n]$ of the polynomials which vanish in all points of $\mathbb{P}^n(\fq)$ is 
generated by $\{ X_j^q X_i - X_i^q X_j\, | 
\, 0 \leq i < j \leq n\}$. We denote by $\fq[X_0, \ldots , X_n]_d$
 (respectively, $(J_q)_d$) the $\fq$-vector subspace formed by the 
homogeneous polynomials of degree $d$ (together with the zero polynomial) in  
$\fq[X_0, \ldots , X_n]$ (respectively, $J_q$).

\begin{definition}
Let $d$ be a positive integers and 
let $\psi: \fq[X_0, \ldots , X_n]_d / (J_q)_d\rightarrow \fq^N$ be the 
$\fq$-linear transformation given by $\psi( f + (J_q)_d ) = (f(Q_1) \ldots, 
f(Q_N) )$, where we write the points of $\mathbb{P}^n(\fq)$ in the standard 
notation, i.e.\ the first nonzero entry from the left is equal to 1. The 
projective Reed-Muller code of order $d$, denoted by $\prm(n, d)$, is the image 
of 
$\psi$.
\end{definition}
In \cite{sorensen} S\o rensen  determined all values for  
the minimum distance $W^{(1)}_{\prm}(n, d)$ of $\prm(n, d)$ and proved that 
\[	
W^{(1)}_{\prm}(n, d) = 
W^{(1)}_{\RM}(n, d-1).
\]

One may wonder if there is a similar relation between the next-to-minimal 
weight of 
projective Reed-Muller codes and the  next-to-minimal weight of 
generalized Reed-Muller codes. A hint that this might be true comes from the 
following reasoning.
Let $\omega$ be the Hamming weight of $\varphi(g + I_q)$, where $g \in \fq[X_1, 
\ldots, X_n]$ is a polynomial of degree $d - 1$, and let $g^{(h)}$ be the 
homogenization of $g$ with respect to $X_0$. Then the degree of $g^{(h)}$ is $d 
- 1$ 
and the weight of $\psi(X_0 g^{(h)} + (J_q)_d)$ is $\omega$. This shows that, 
denoting by $W^{(2)}_{\prm}(n, d)$ the next-to-minimal weight of $\prm(n, d)$,  
we 
have 
\begin{equation}\label{2.1}
W^{(2)}_{\prm}(n, d) \leq W^{(2)}_{\RM}(n, d - 1).
\end{equation}

In \cite{carv-neu-ieee} we determined all values for the next-to-minimal 
weights of projective Reed-Muller codes defined over $\mathbb{F}_2$, and from 
those results we get that there are cases for which equality does not hold in 
\eqref{2.1}.
In the next section 
we 
will 
determine several values for the 
next-to-minimal 
weights of projective Reed-Muller codes defined over $\mathbb{F}_q$, with $q 
\geq 3$, and we will also find some cases where equality does not hold in 
\eqref{2.1}. We recall from \cite{carv-neu-ieee} a definition and some 
results that we will need to prove the main results.

\begin{definition}
Let $f \in \fq[X_0, \ldots , X_n]_d$. The set of points of $\mathbb{P}^n(\fq)$ 
which are not zeros of $f$ is called the support of $f$, and we denote its 
cardinality by $| f |$ (hence $|f|$ is the weight of the codeword $\psi(f + 
(J_q)_d)$).
\end{definition}

In what follows the integers $k$ and $\ell$ will always be the ones uniquely 
defined 
by the equality   
\[ 
d - 1 = k (q - 1) + \ell
\] 
with $0 \leq k \leq n - 1$ and $0 < \ell \leq q - 1$.

\begin{theorem} (\cite[Lemma 2.4, Lemma 2.5, Prop. 2.6 and Prop. 
2.7]{carv-neu-ieee}) 
\label{old-results} Let 
$f \in \fq[X_0,\ldots , X_n]_d$ be a nonzero polynomial, 
and let $S$ be it support, which we assume to be nonempty. Then:\\
i) if there exists a hyperplane $H \subset \mathbb{P}^n(\fq)$
such that $S \cap H = \emptyset$ and $|f| > W^{(1)}_{\prm}(n, d)$ then 
$|f| \geq W^{(2)}_{\RM}(n, d - 1)$; \\
ii) if  $|S| <
	\left(
	1 + \dfrac{1}{q}
	\right)
	(q  - \ell) q^{n-k-1} 
	$
then there  exists a hyperplane $H \subset \mathbb{P}^n(\fq)$  such that $S 
\cap H = 
\emptyset$; \\
iii) if $
|S| \le 
\left(1 + \dfrac{1}{(q-\ell)}  \right)
(q - \ell)q^{n-k-1} = (q - \ell + 1)q^{n-k-1}$ then there exists  $r \ge k$ and 
a linear subspace 
$H_r \subset \mathbb{P}^n(\fq)$ of dimension  $r$ such that $S \cap H_r = 
\emptyset$.
\end{theorem} 

\section{Main results}

In this section we determine the next-to-minimal weight for most cases of  
projective Reed-Muller codes. Recall that we are assuming $q \geq 3$. We start 
by treating the case where $k = n - 1$.

%
%
\begin{proposition}\label{lema9}
If  $d - 1 = (n - 1) (q - 1) + \ell$, with $0 < \ell \leq q   - 
1$ then 
$$
W^{(2)}_{\prm}(n, d) = W^{(2)}_{\RM}(n, d - 1) = q - \ell + 1\, .
$$
\end{proposition}

%
%
\begin{proof}
Let $f\in \fq[X_0,\ldots , X_n]_d$ be a nonzero polynomial such that 
$$
0 \neq |f| \le W^{(2)}_{\RM}(n, d - 1) = q - \ell + 1
\, 
$$
(such a polynomial exists because $W^{(1)}_{\prm}(n, d ) = W^{(1)}_{\RM}(n, d - 
1) = q - \ell$). Let $S$ be the support of $f$, from Theorem \ref{old-results} 
(iii) 
there exists a hyperplane $H$ such that $S \cap H = 
\emptyset$. From Theorem \ref{old-results} (i) we get that 
$|f|= W^{(1)}_{\prm}(n, d )$ or $|f| \ge W^{(2)}_{\RM}(n, d - 1)$, so that  
$W^{(2)}_{\prm}(n, d) \geq W^{(2)}_{\RM}(n, d - 1)$, and from inequality 
\eqref{2.1} we get $W^{(2)}_{\prm}(n, d) = W^{(2)}_{\RM}(n, d - 1)$.
\end{proof}
%
%

%
%
\noindent
Now we start to study the case where $k < n - 1$.

\begin{proposition}\label{prop32}
Let  $d - 1 = k (q - 1) + \ell$  be such that 
$k < n - 1$. If either 
$1 < \ell \le \dfrac{q+1}{2}$,  or 
$q=3$, $k > 0$ and $\ell=1$, we get
$$
W^{(2)}_{\prm}(n, d) = W^{(2)}_{\RM}(n, d - 1) = 
\left\{ 
\begin{matrix}
(q- 1)(q - \ell + 1)q^{n-k-2} 
      & \text{ if\ }\; 1 < \ell \le \dfrac{q+1}{2}\, , \\
 8 \cdot 3^{n-k-2}
    & \text{ if\ }\; q = 3, k > 0 \text{ and } \ell = 1\, .
\end{matrix}
\right.
$$

\end{proposition}
%
%
\begin{proof}
From \eqref{table-2ndwt} we 
get that
$$
W^{(2)}_{\RM}(n, d - 1) 
 =  \left(1 + \dfrac{c}{(q-\ell)q}  \right)
(q - \ell)q^{n-k-1} \, , \
$$
where, when $q = 3$, $0 < k \leq n - 2$  and $\ell = 1$ we have $c = q - 1 = 2$ 
(and a fortiori 
$\dfrac{c}{(q-\ell)q}= \dfrac{1}{q}$), and when $q \geq 3$ 
and $1 < \ell \leq 
(q + 1)/2$ 
we have
$c = \ell - 1$
(and a fortiori 
$\dfrac{c}{(q-\ell)q} \leq \dfrac{1}{q}$).
Assume, by means of absurd, that there exists  $f\in \fq[X_0,\ldots , X_n]_d$ 
such that 
$$
W^{(1)}_{\prm}(n, d ) < |f| < W^{(2)}_{\RM}(n, d - 1)\leq \left(1 + 
\dfrac{1}{q}  
\right)
(q - \ell)q^{n-k-1}
\,  
$$
and let $S$ be the support of $f$.  From Theorem \ref{old-results} (ii) and (i) 
we get 
$|f| \geq W^{(2)}_{\RM}(n, d - 1)$, a contradiction. Hence 
$W^{(2)}_{\prm}(n, d )  \geq W^{(2)}_{\RM}(n, d - 1)$ and a fortiori 
$W^{(2)}_{\prm}(n, d ) = W^{(2)}_{\RM}(n, d - 1)$.
\end{proof}

Next we treat the case where $n \ge 3$,
$0 \le k < n - 2$ and $\ell =1$. We observe from \eqref{table-2ndwt} 
that $(q^2-1)q^{n-k-2} \leq W^{(2)}_{\RM}(n, d - 1)$ (actually we have equality 
when $q = 3$ and an strict inequality when $q > 3$).

\begin{proposition}\label{prop33}
Let  $n \ge 3$,
$0 \le k < n - 2$ and $\ell =1$,  then
$$
W_{\prm}^{(2)} (n,d) = 
\left(1 + \frac{1}{q} \right) (q-1)q^{n-k-1} = 
(q^2-1)q^{n-k-2}\, .
$$
\end{proposition}
\begin{proof}
In case $1 \leq k < n - 2$ \/ let  
$f=X_1X_{k+3}g + X_0X_{k+2}h$, where
$$
g = \prod_{i=2}^{k+1} \left( X_{i}^{q-1} - X_{1}^{q-1} \right)
\quad \text{and} \quad 
h = \prod_{i=2}^{k+1} \left( X_{i}^{q-1} - X_{0}^{q-1} \right)
\,, 
$$
 or let $f = X_1X_{3} + X_0X_{2}$ in case $k = 0$. 
We claim that $ | f | = (q^2-1)q^{n-k-2}$, and
let $\balpha = (\alpha_0: \ldots : \alpha_n)$ be a point in the support of $f$.
If $\alpha_0 = 0$ then we may take $\alpha_1 = 1$, and we have $f(\balpha) = 
\alpha_{k + 3} \prod_{i = 2}^{k + 1} (\alpha_i^{q - 1} - 1)$ (or $f(\balpha) = 
\alpha_3$ in case $k = 0$). Thus we must have $\alpha_{k+3} \in \fq^*$ and 
$\alpha_2 = \cdots = \alpha_{k + 1} = 0$, so
we get $(q - 1)q^{n - k - 2}$ points of this form in the support. If 
$\alpha_0 = 1$ then 
\[
f(\balpha) = \alpha_1 \alpha_{k + 3} \prod_{i = 2}^{k + 1} 
(\alpha_i^{q - 1} - \alpha_1^{q - 1}) + \alpha_{k + 2} \prod_{i = 2}^{k + 1} 
(\alpha_i^{q - 1} - 1)
\]
in case $1 \leq k < n - 2$, or $f(\balpha) = \alpha_1 \alpha_{3} + \alpha_2$ 
in case $k= 0$. We see that for any $\alpha_1 \in \fq$ we get 
$f(\balpha) = (\alpha_1 \alpha_{k + 3}+ \alpha_{k + 2}) \prod_{i = 2}^{k + 1} 
(\alpha_i^{q - 1} - 1)$, thus we must have $\alpha_{k + 2} \neq -\alpha_1 
\alpha_{k + 3}$ and  $\alpha_2 = \cdots = \alpha_{k + 1} = 0$, so we get $(q - 
1) q^{n - k - 1}$ points of this form in the support, and we get
\[
| f | = (q - 1)q^{n - k - 2} + (q - 
1) q^{n - k - 1} = (q^2-1)q^{n-k-2},
\]
this proves that $W_{\prm}^{(2)} (n,d) \leq (q^2-1)q^{n-k-2}$.   
Assume, by means of absurd, that there exists  $a\in \fq[X_0,\ldots , X_n]_d$ 
such that 
$$
W^{(1)}_{\prm}(n, d ) < |a| < (q^2-1)q^{n-k-2} = \left(1 + 
\dfrac{1}{q}  
\right)
(q - 1)q^{n-k-1} \leq W^{(2)}_{\RM}(n, d - 1)
\,  
$$
and let $S$ be the support of $a$.  From Theorem \ref{old-results} (ii) and (i) 
we get 
$|a| \geq W^{(2)}_{\RM}(n, d - 1)$, a contradiction, and we conclude that   
$W^{(2)}_{\prm}(n, d ) = (q^2-1)q^{n-k-2}$.
\end{proof}

The above proof shows that the next-to-minimal weight of projective Reed-Muller 
codes is not, in all cases, attained only by evaluating polynomials that 
split completely as  a product of degree one factors. For example, in the case 
$k = 0$ of the above result (hence $n = 2$ and $d = 2$) the next-to-minimal 
weight is attained from the  evaluation of an irreducible quadric. This is in 
contrast with what happens with polynomials that attain the minimum distance of 
$\RM(n, d)$ and $\prm(n, d)$ which must be product of degree one polynomials 
(see \cite{delsarte-et-al} and \cite{rolland2} respectively). Also, in 
\cite{leduc}, the author proves that  
 codewords of next-to-minimal weight in $\RM(n, d)$, when $q \geq 
3$, always come from the evaluation of polynomials which are products of degree 
one polynomials. 

The following result deals with the case where $n = 2$ and $d = 2$ (so that $k 
= 0$ and $\ell = 1$).

\begin{proposition}
Let $n = 2$ and $d=2$ then
$$
W^{(2)}_{\prm}(2,2) = q^2 .
$$
\end{proposition}

%
%
\begin{proof}
Let $g \in \fq[X_0, X_1 , X_2]_2$ be such that  
\[
W^{(1)}_{\prm}(2,2) < | g | < 
(1 + \frac{1}{q}) 
W^{(1)}_{\prm}(2,2) = q^2 - 1.
\]
From Theorem \ref{old-results} (ii) and (i) 
we get 
$|g| \geq W^{(2)}_{\RM}(2, 1) = q^2$, a contradiction, hence from $\eqref{2.1}$ 
we get 
$W^{(2)}_{\prm}(2,2) \in \{q^2 - 1, q^2\}$. 
Let $f \in \fq[X_0, X_1 , X_2]_2$ be such that  
$| f | \in \{q^2 - 1, q^2\}$ and let $S$ be its support.  From Theorem 
\ref{old-results} (i) if there 
exists a line in $\mathbb{P}^2(\fq)$ not intersecting $S$ 
then $| f | = q^2$, so let's assume that there is no such line. After a 
projective transformation we may assume that $P = (0:0:1)$ is not in $S$, and 
let $L_0, \ldots, L_q$ be the lines that contain $P$. Let $i \in \{0, \ldots, 
q\}$,  since $\deg(f) = 2$ we have $| S \cap L_i | \geq q - 1$ 
and from $| S | \in \{ q^2 - 1, q^2\}$ we must have at most one line 
intersecting $S$ in $q$ points. After a projective transformation that fixes 
$P$ we may assume that the line $X_0 = 0$ intersects $S$ in $q - 1$ points, and 
that $P$ and $Q = (0:1:0)$ are the points of the line missing $S$, so we may 
write $f= X_0 (a_0 X_0 + a_1 X_1 + a_2 X_2) + X_1 X_2$. Observe that there are 
$q - 1$ points of the form $(0:1:\alpha)$ in the support of $f$. Let $\alpha, 
\beta \in \fq$, then $f(1:\alpha:\beta) = a_0 + a_1 \alpha + (a_2 + 
\alpha)\beta$ and for each $\alpha \neq - a_2$ we have $q - 1$ values for 
$\beta$ such that  $f(1:\alpha:\beta) \neq 0$. On the other hand, if 
$\alpha = - a_2$ then either $f(1: - a_2:\beta) = 0$ or $f(1: - a_2:\beta) \neq 
0$ for $\beta \in \fq$, so that  $|f| = (q-1) + (q-1)^2 = (q-1)q$ or $|f| = 
(q-1) + (q-1)^2 + q = q^2$, which proves $|f|  = q^2$. 
\end{proof}

We summarize the results for $W_{\prm}^{(2)} (n,d)$ obtained in  
\cite{carv-neu-ieee} and in this paper in the 
following tables, where we also list the corresponding values  
of $W_{\RM}^{(2)} (n,d-1)$ for comparison.\\ \\

\begin{table}[htbp]
	\centering
	\begin{tabular}{ccccc}\hline\hline
		$n$ & $k$   & $\ell$ 
		& $W_{\RM}^{(2)} (n,d-1)$ & $W_{\prm}^{(2)} (n,d)$ 
		 \\
		\hline\hline
		$n \ge 3 $    & $k=0$ & $\ell =1$
		&	$2^n$	 & $3 \cdot 2^{n-2}$  \\
		$n\ge 4$ & $1 \le k < n-2$ & $\ell =1$
		&	$3\cdot 2^{n-k-2}$	 & $3 \cdot 2^{n-k-2}$\\
		$n\ge 2$ & $k=n-2$ & $\ell =1$
		&	$4$	 & $4$  \\
        $n\ge 2$ & $k=n-1$ & $\ell =1$
        		&	$2$	 & $2$  \\
		\hline\hline
	\end{tabular}
	\caption{Next-to-minimal weights for $\RM(n, d)$ and $\prm(n, d)$ when  
	$n\ge 2$ and $q=2$}
	\label{pesoq2}
\end{table}

\begin{table}[htbp]
	\centering
	\begin{tabular}{ccccc}\hline\hline
		$n$ & $k$   & $\ell$ 
		& $W_{\RM}^{(2)} (n,d-1)$ & $W_{\prm}^{(2)} (n,d)$ \\
		\hline\hline
		$n=2$ & $k=0$ & $\ell =1$
		&	$3^2$	 & $3^2$   \\
		$n \ge 3$ & $k=0$ & $\ell =1$
		&	$3^n$	 & $8 \cdot 3^{n-2}$   \\
		$n \ge 3$ & $1 \le k \le n-2$ & $\ell =1$
		&	$8\cdot 3^{n-k-2}$	 & $8\cdot 3^{n-k-2}$ \\
		$n \ge 2$ &  $0 \le k \le n-2$ & $\ell =2$
		&	$4\cdot 3^{n-k-2}$	 & $4\cdot 3^{n-k-2}$ \\
		$n \ge 1$ &  $k=n-1$ & $\ell =1, 2$
		&	$4 - \ell$	 &  $4 - \ell$ \\
		\hline\hline
	\end{tabular}
	\caption{Next-to-minimal weights for $\RM(n, d)$ and $\prm(n, d)$ when  
	$n\ge 1$ and $q=3$}
	\label{pesoq3}
\end{table}

\begin{table}[htbp]
	\centering
	\begin{tabular}{ccccc}\hline\hline
		$n$ & $k$   & $\ell$ 
		& $W_{\RM}^{(2)} (n,d-1)$ & $W_{\prm}^{(2)} (n,d)$ \\
		\hline\hline
		$n=2$ & $k=0$ & $\ell =1$
		&	$q^2$	 & $q^2$ \\
		$n\ge 3$ & $k < n-2$ & $\ell =1$
		&	$q^{n-k}$	 & $q^{n-k} - q^{n-k-2}$\\
		$n \ge 3$ & $k = n -2$ & $\ell =1$
		&	$q^2$	 & ??? \\
		$n \ge 2$ & $k \le n -2$ & $1 < \ell \le \frac{q+1}{2}$
		&	$(q-1)(q-\ell 
		+ 1)q^{n-k-2}$	 &  $(q-1)(q-\ell + 1)q^{n-k-2}$ \\
		$n \ge 2$ & $k \le n -2$ & $\frac{q+1}{2} < \ell \leq q -1 $ 
		&	$(q-1)(q-\ell + 1)q^{n-k-2}$	 & ???  \\
		$n \ge 1$ & $k = n -1$ & $1 \leq \ell \leq q-1$
		&	$q-\ell + 1$	 &   $q-\ell + 1$ \\
		\hline\hline
	\end{tabular}
	\caption{Next-to-minimal weights for $\RM(n, d)$ and $\prm(n, d)$ when  
	$n\ge 1$ and $q \geq 4$}
	\label{pesonovo}
\end{table} 
\newpage
\     

\bibliographystyle{plain}

\begin{thebibliography}{10}




\bibitem{rolland-ballet} S.\ Ballet, R.\ Rolland, On low weight codewords of 
generalized affine
and projective Reed-Muller codes, Des. Codes Cryptogr. \textbf{73} (2014) 
271--297.



%

\bibitem{bruen}
A.\ Bruen, Blocking sets and low-weight codewords in the generalized 
Reed-Muller codes, Contemp. Math. \textbf{525} (2010) 161--164.


%
%






\bibitem{carv-neu-ieee} C.\ Carvalho, C.; V.\ G.L.\ Neumann,  The 
next-to-minimal 
weights of binary projective Reed-Muller codes. IEEE Transactions on 
Information Theory, \textbf{62} (2016) 6300--6303. 





\bibitem{erickson} D. Erickson,  Counting zeros of polynomials over finite fields. PhD Thesis, California Institute of Technology, Pasadena (1974).


\bibitem{delsarte-et-al} P.\ Delsarte, J.\  Goethals, F.\ MacWilliams, On 
generalized Reed-Muller codes and their relatives, Inform.
Control \textbf{16} (1970) 403--442.

\bibitem{geil} O. Geil,
On the second weight of generalized Reed-Muller codes.
Des.\ Codes Cryptogr. \textbf{48}(3) (2008) 323--330.



\bibitem{kasami} T. Kasami, S. Lin, and W. W. Peterson, New generalizations of the Reed-Muller codes-Part  I: Primitive codes, IEEE Trans. Inform. 
Theory, \t2xtbf{14} (1968) 189--199. 



%
%
%


\bibitem{lachaud} G. Lachaud,
The parameters of projective Reed-Muller codes,
Discrete Math. \textbf{81}(2) (1990) 217--221.

\bibitem{leduc} E.\ Leducq, Second weight codewords of generalized Reed-Muller 
codes,  Cryptogr.\ Commun.\ \textbf{5} (2013) 241--276.






\bibitem{mercier-rolland} D.-J. Mercier, R.\ Rolland, Polyn\^omes homog\`enes 
qui s'annulent sur l'espace projectif $\mathbb{P}^m(\fq)$, J. Pure
Appl. Algebra \textbf{124} (1998) 227--240.


\bibitem{muller} D.E.\ Muller,  Application of boolean algebra to switching circuit design and to error detection,  IRE Transactions on Electronic Computers \textbf{3} (1954) 6-12.


\bibitem{reed} I.S.\ Reed, A class of multiple-error-correcting codes and the decoding scheme, Transactions of the IRE Professional Group on Information Theory, \textbf{3} (1954) 38-49.

\bibitem{idealapp} C.\ Renter\'\i a, H.\ Tapia-Recillas, 
Reed-Muller codes: an ideal theory approach,
Commu. Algebra \textbf{25}(2) (1997) 401--413.


\bibitem{rodier-sboui} F.\ Rodier, A.\  Sboui A., Les arrangements minimaux et 
maximaux d'hyperplans dans $\mathbb{P}^n(\mathbb{F}_q)$, C. R. Math.
Acad. Sci. Paris \textbf{344} (2007) 287--290.

\bibitem{rodier-sboui2}  F.\ Rodier, A.\  Sboui A.,
Highest numbers of points of hypersurfaces over finite fields and generalized Reed-Muller codes,  Finite Fields Appl. \textbf{14} (2008) 816--822.


\bibitem{rolland} R.\ Rolland, The second weight of generalized Reed-Muller
codes in most cases. Cryptogr.\ Commun.\ \textbf{2} (2010) 19--40.


\bibitem{rolland2} R.\ Rolland,  Number of points of non-absolutely irreducible 
hypersurfaces. Algebraic geometry and its applications, 481--487, Ser.\ Number 
Theory Appl. 5, World Sci. Publ., Hackensack, NJ, 2008.


\bibitem{sboui1} A. Sboui,   Special numbers of rational points on hypersurfaces in the $n$-dimensional projective space over
a finite field, Discret. Math. \textbf{309} (2009) 5048-?5059. 

\bibitem{serre} J.P.\ Serre, Lettre \`{a} M. Tsfasman du 24 Juillet 1989. In: Journ\'ees arithm\'etiques de Luminy 17--21 Juillet
1989, Ast\'erisque, 198-?200. Soci\'et\'e Math\'ematique de France (1991).




\bibitem{sorensen} A. S{\o}rensen, Projective Reed-Muller codes, 
IEEE Trans. Inform. Theory \textbf{37}(6) (1991) 1567--1576.






\end{thebibliography}

\end{document}